\documentclass[journal]{IEEEtran}
\usepackage[latin9]{inputenc}
\usepackage{amsthm}
\usepackage{amsmath}
\usepackage{amssymb}
\usepackage{graphicx}

\makeatletter
\theoremstyle{plain}
\newtheorem{thm}{\protect\theoremname}
\theoremstyle{plain}
\newtheorem{lem}[thm]{\protect\lemmaname}

\usepackage{amsfonts}
\providecommand{\U}[1]{\protect\rule{.1in}{.1in}}
\newtheorem{theorem}{Theorem}

\makeatother

\providecommand{\lemmaname}{Lemma}
\providecommand{\theoremname}{Theorem}

\begin{document}

\title{A Convex Framework for Optimal Investment on Disease Awareness in
Social Networks}

\author{Victor M. Preciado, Faryad Darabi Sahneh, and Caterina Scoglio}
\maketitle
\begin{abstract}
We consider the problem of controlling the propagation of an epidemic
outbreak in an arbitrary network of contacts by investing on disease
awareness throughout the network. We model the effect of agent awareness
on the dynamics of an epidemic using the SAIS epidemic model, an extension
of the SIS epidemic model that includes a state of ``awareness''.
This model allows to derive a condition to control the spread of an
epidemic outbreak in terms of the eigenvalues of a matrix that depends
on the network structure and the parameters of the model. We study
the problem of finding the cost-optimal investment on disease awareness
throughout the network when the cost function presents some realistic
properties. We propose a convex framework to find cost-optimal allocation
of resources. We validate our results with numerical simulations in
a real online social network. 
\end{abstract}

\section{Introduction}

Development of strategies to control spreading processes is a central
problem in public health and network security \cite{Bai75}. Motivated
by the problem of epidemic spread in human networks, we analyze the
problem of controlling the spread of a disease by investing on disease
awareness throughout the individuals in a social network. The dynamics
of the spread depends on both the structure of the network of contacts,
the epidemic model and the values of the parameters associated to
each individual. We model the spread using a recently proposed variant
of the popular SIS (Susceptible-Infected-Susceptible) epidemic model
that includes a state of ``awareness'' (or ``alertness''), \cite{FaryadCDC11SAIS}.
In our setting, we can modify the individual levels of awareness,
within a feasible range, by investing resources in each node. These
investments have associated costs, which can vary from individual
to individual. In this context, we propose an efficient convex framework
to find the cost-optimal investment on awareness in a social network
to control an epidemic outbreak.

The paper is organized as follows. In Section II, we introduce our
notation, as well as some background needed in our derivations. In
Section III, we formulate our problem and provide an efficient solution
based on convex optimization. In Section IV, we show simulations supporting
our results.

\section{\label{Notation}Notation \& Preliminaries}

In this section, we introduce some graph-theoretical nomenclature
and the epidemic spreading model under consideration.

\subsection{Graph Theory}

Let $\mathcal{G}=\left(\mathcal{V},\mathcal{E}\right)$ denote an
undirected graph with $n$ nodes, $m$ edges, and no self-loops%
\footnote{An undirected graph with no self-loops is also called a \emph{simple}
graph.%
}. We denote by $\mathcal{V}\left(\mathcal{G}\right)=\left\{ v_{1},\dots,v_{n}\right\} $
the set of nodes and by $\mathcal{E}\left(\mathcal{G}\right)\subseteq\mathcal{V}\left(\mathcal{G}\right)\times\mathcal{V}\left(\mathcal{G}\right)$
the set of undirected edges of $\mathcal{G}$. If $\left\{ i,j\right\} \in\mathcal{E}\left(\mathcal{G}\right)$
we call nodes $i$ and $j$ \emph{adjacent} (or neighbors), which
we denote by $i\sim j$. We define the set of neighbors of a node
$i\in\mathcal{V}$ as $\mathcal{N}_{i}=\{j\in\mathcal{V}\left(\mathcal{G}\right):\left\{ i,j\right\} \in\mathcal{E}\left(\mathcal{G}\right)\}$.
The number of neighbors of $i$ is called the \emph{degree} of node
$i$, denoted by $d_{i}$. The adjacency matrix of an undirected graph
$\mathcal{G}$, denoted by $A_{\mathcal{G}}=[a_{ij}]$, is an $n\times n$
symmetric matrix defined entry-wise as $a_{ij}=1$ if nodes $i$ and
$j$ are adjacent, and $a_{ij}=0$ otherwise%
\footnote{For simple graphs, $a_{ii}=0$ for all $i$.%
}. Since $A_{\mathcal{G}}$ is symmetric, all its eigenvalues, denoted
by $\lambda_{1}(A_{\mathcal{G}})\geq\lambda_{2}(A_{\mathcal{G}})\geq\ldots\geq\lambda_{n}(A_{\mathcal{G}})$,
are real.

\subsection{Heterogenous SAIS Model\label{Sec-Problem Statement}}

The Susceptible-Alert-Infected-Susceptible (SAIS) model proposed in
\cite{FaryadCDC11SAIS} extends the SIS epidemic model by incorporating
the response of individuals in the dynamics of the infection. In this
paper, we further extend the SAIS model by considering heterogenous
transition rates. The contact topology in this formulation is an arbitrary,
generic graph $\mathcal{G}$ with nodes representing individuals and
links representing contacts among them. Each node is allowed to be
in one of the three states \emph{`susceptible'}, \emph{`infected'},
and \emph{`alert'} (also called \emph{`aware'}). For each agent $i\in\{1,...,N\}$,
we let the random variable $x_{i}(t)=e_{1}$ if the agent $i$ is
susceptible at time $t$, $x_{i}(t)=e_{2}$ if alert, and $x_{i}(t)=e_{3}$
if infected, where $e_{1}=[1,0,0]^{T}$, $e_{2}=[0,1,0]^{T}$, and
$e_{3}=[0,0,1]^{T}$ are the standard unit vectors of $\mathbb{R}^{3}$.
If agent $i$ and agent $j$ are in contact $a_{ij}=1$, otherwise
$a_{ij}=0$. In heterogenous SAIS\ model, four possible types of
stochastic transitions determine an agent's state at each time:
\begin{enumerate}
\item Susceptible agent $i$ becomes infected by the infection rate $\beta_{i}\in\mathbb{R}^{+}$
times the number of its infected neighbors, i.e., 
\begin{multline}
\Pr[x_{i}(t+\Delta t)=e_{3}|x_{i}(t)=e_{1},X(t)]=\\
\beta_{i}Y_{i}(t)\Delta t+o(\Delta t),
\end{multline}
for $i\in\{1,...,N\}$, and $Y_{i}(t)\triangleq\sum_{j\in\mathcal{N}_{i}}a_{ij}1_{\{x_{j}(t)=e_{1}\}}$
(the number of infected neighbors of agent $i$ at time $t$).
\item Infected agent $i$ recovers back to susceptible state by the curing
rate $\delta_{i}\in\mathbb{R}^{+}$, i.e.,
\begin{equation}
P(x_{i}(t+\Delta t)=e_{1}|x_{i}(t)=e_{3},X(t))=\delta_{i}\Delta t+o(\Delta t).
\end{equation}

\item Susceptible agents might go to the alert state if surrounded by infected
individuals. Specifically, a susceptible node becomes alert with the
alerting rate $\kappa_{i}\in\mathbb{R}^{+}$ times the number of its
infected neighbors, i.e., 
\begin{multline}
P(x_{i}(t+\Delta t)=e_{2}|x_{i}(t)=e_{1},X(t))=\\
\kappa_{i}Y_{i}(t)\Delta t+o(\Delta t).
\end{multline}

\item An alert agent can get infected in a process similar to a susceptible
agent but with a smaller infection rate $r_{i}\beta_{i}$ with $0<r_{i}<1$,
i.e.,
\begin{multline}
P(x_{i}(t+\Delta t)=e_{3}|x_{i}(t)=e_{2},X(t))=\\
r_{i}\beta_{i}Y_{i}(t)\Delta t+o(\Delta t).\label{AlertingProcess}
\end{multline}

\end{enumerate}
In above equations, $\Pr[\cdot]$ denotes probability, $X(t)\triangleq\{x_{i}(t),i=1,...,N\}$
is the joint state of the network, $\Delta t>0$ is a time step, and
the indicator function $1_{\{\mathcal{X}\}}$ is one if $\mathcal{X}$\ is
true and zero otherwise.

Let $p_{i}$ and $q_{i}$ denote the probabilities of agent $i$ to
be infected and alert, respectively. Using a first-order mean-field
approximation \cite{Sahneh2013TON}, the heterogeneous SAIS\ model
is obtained as:
\begin{align}
\dot{p}_{i} & =\beta_{i}(1-p_{i}-q_{i})\sum_{j\in\mathcal{N}_{i}}a_{ij}p_{j}+r_{i}\beta_{i}q_{i}\sum_{j\in\mathcal{N}_{i}}a_{ij}p_{j}-\delta_{i}p_{i},\label{dp}\\
\dot{q}_{i} & =\kappa_{i}(1-p_{i}-q_{i})\sum_{j\in\mathcal{N}_{i}}a_{ij}p_{j}-r_{i}\beta_{i}q_{i}\sum_{j\in\mathcal{N}_{i}}a_{ij}p_{j},\label{dq}
\end{align}
for $i\in\{1,...,N\}$.

\subsection{Stability Analysis of Heterogenous SAIS Model\label{Sec-Main}}

The sufficient condition for exponential stability of heterogeneous
SIS model in \cite{PreciadoCDC13} is also a sufficient condition
for exponential die out of initial infections in heterogeneous SAIS
model (\ref{dp})-(\ref{dq}). Including an state of alertness in
the network induces a secondary dynamics that can potentially control
the spread of the disease. The following theorem states a sufficient
condition for the\emph{ }epidemics to die out in the heterogenous
SAIS\ model.

\begin{theorem}\label{SAIS_Stab} The infection probabilities of
the heterogenous SAIS model in (\ref{dp}) and (\ref{dq}) tend to
zero, i.e., $\lim_{t\to\infty}p_{i}\left(t\right)=0$, if
\begin{equation}
\lambda_{1}(LBA_{\mathcal{G}}-MD)<0,\label{dieoutcond}
\end{equation}
where $B\triangleq\mbox{diag}\{\beta_{i}\}$, $D\triangleq\mbox{diag}\{\delta_{i}\}$,
$L\triangleq diag\{r_{i}\bar{\kappa}_{i}+r_{i}\},~M\triangleq diag\{\bar{\kappa}_{i}+r_{i}\},~\bar{\kappa}_{i}\triangleq\frac{\kappa_{i}}{\beta_{i}}.$

\end{theorem}

\begin{proof} Denoting steady state values of infection and alertness
probabilities of agent $i$ with $p_{i}^{\ast}$ and $q_{i}^{\ast}$,
respectively, we find from (\ref{dq})
\begin{equation}
q_{i}^{\ast}\sum a_{ij}p_{j}^{\ast}=\frac{\bar{\kappa}_{i}}{\bar{\kappa}_{i}+r_{i}}(1-p_{i}^{\ast})\sum a_{ij}p_{j}^{\ast}.
\end{equation}
Replacing for $q_{i}^{\ast}\sum a_{ij}p_{j}^{\ast}$ in (\ref{dp})
and some straightforward algebra, $p_{i}^{\ast}$ is found to satisfy
\begin{equation}
\frac{p_{i}^{\ast}}{1-p_{i}^{\ast}}=\frac{\beta_{i}}{\delta_{i}}(r_{i}\frac{\bar{\kappa}_{i}+1}{\bar{\kappa}_{i}+r_{i}})\sum a_{ij}p_{j}^{\ast}.\label{piss}
\end{equation}
Healthy state $p_{i}^{\ast}=0$ for $i\in\{1,...,N\}$ is always a
solution to (\ref{piss}). However, the healthy equilibrium is not
necessarily stable. Indeed, in order to find the die-out condition,
we seek existence of nontrivial solutions of the equilibrium equation
(\ref{piss}). Assume that the infection rate $\beta_{i}$ can be
decomposed $\beta_{i}=\theta\bar{\beta}_{i}$, so that auxiliary parameter
$\theta$ can globally vary the values of infection rate of all agents.
The idea of the remaining is to use bifurcation analysis with respect
to $\theta$ for 
\begin{equation}
\frac{p_{i}^{\ast}}{1-p_{i}^{\ast}}=\theta\frac{\bar{\beta}_{i}}{\delta_{i}}(r_{i}\frac{\bar{\kappa}_{i}+1}{\bar{\kappa}_{i}+r_{i}})\sum a_{ij}p_{j}^{\ast}.\label{pisstheta}
\end{equation}

Since by definition $\frac{\beta_{i}}{\delta_{i}}(r_{i}\frac{\bar{\kappa}_{i}+1}{\bar{\kappa}_{i}+r_{i}})>0$,
(\ref{piss}) implies that if contact graph $\mathcal{G}$ is connected,
either infection probabilities of all agents are zero or they are
all absolutely positive. Therefore, a critical value $\theta_{c}$
for $\theta$ must exists such that for $\theta=\theta_{c}$,
\begin{equation}
p_{i}^{\ast}=0,\frac{dp_{i}^{\ast}}{d\theta}|_{\theta=\theta_{c}}>0.
\end{equation}
Taking derivative of both sides of (\ref{pisstheta}) w.r.t $\theta$
yields
\begin{align}
\frac{1}{(1-p_{i}^{\ast})^{2}}\frac{dp_{i}^{\ast}}{d\theta} & =\frac{\bar{\beta}_{i}}{\delta_{i}}(r_{i}\frac{\bar{\kappa}_{i}+1}{\bar{\kappa}_{i}+r_{i}})\sum a_{ij}p_{j}^{\ast}\nonumber \\
 & +\theta\frac{\bar{\beta}_{i}}{\delta_{i}}(r_{i}\frac{\bar{\kappa}_{i}+1}{\bar{\kappa}_{i}+r_{i}})\sum a_{ij}\frac{dp_{j}^{\ast}}{d\tau}.
\end{align}
Replacing for $p_{i}^{\ast}=0$ at $\theta=\theta_{c}$ reduces to
\begin{equation}
\frac{dp_{i}^{\ast}}{d\theta}|_{\theta=\theta_{c}}=\theta_{c}\frac{\bar{\beta}_{i}}{\delta_{i}}(r_{i}\frac{\bar{\kappa}_{i}+1}{\bar{\kappa}_{i}+r_{i}})\sum a_{ij}\frac{dp_{j}^{\ast}}{d\theta}|_{\theta=\theta_{c}}.\label{threshold_eq}
\end{equation}
Hence, the critical value $\theta_{c}$ is such that (\ref{threshold_eq})
has a positive solution for $\frac{dp_{i}^{\ast}}{d\theta}|_{\theta=\theta_{c}}$.
Equation (\ref{threshold_eq}) is actually a Perron-Frobenius problem
\begin{equation}
w=\theta_{c}HAw,
\end{equation}
where
\begin{equation}
w\triangleq\lbrack\frac{dp_{1}^{\ast}}{d\theta}|_{\theta=\theta_{c}},...,\frac{dp_{N}^{\ast}}{d\theta}|_{\theta=\theta_{c}}],
\end{equation}
and
\begin{equation}
H\triangleq diag\{\frac{\bar{\beta}_{i}}{\delta_{i}}(r_{i}\frac{\bar{\kappa}_{i}+1}{\bar{\kappa}_{i}+r_{i}})\},
\end{equation}
which is guaranteed to have a positive solution $w>0$ with $\theta_{c}=\dfrac{1}{\lambda_{1}(HA)}$
if contact graph is connected. Therefore, for $\theta<\theta_{c}$
healthy state is the only equilibrium while for $\theta>\theta_{c}$
a non-healthy equilibrium exists. Equivalently, rewriting (\ref{threshold_eq})
as
\begin{equation}
(\bar{\kappa}_{i}+r_{i})\delta_{i}\frac{dp_{i}^{\ast}}{d\theta}|_{\theta=\theta_{c}}=\theta_{c}\bar{\beta}_{i}r_{i}(\bar{\kappa}_{i}+1)\sum a_{ij}\frac{dp_{j}^{\ast}}{d\theta}|_{\theta=\theta_{c}},
\end{equation}
suggests that under (\ref{dieoutcond}) healthy state is the only
possible equilibrium. \end{proof}

\section{Optimal Investment in Disease Awareness}

Imagine now that we want to start an awareness campaign to control
the spread of an epidemic outbreak in a social network: On what nodes
in the network should we invest our efforts in order to contain the
spread in a cost-optimal manner? To answer this question, we develop
in this section an optimization framework to find the optimal distribution
of resources to control the values of the awareness rates in the network,
$\left\{ \kappa_{i}\right\} _{i\in\mathcal{V}}$, assuming there is
a known cost function $f_{i}\left(\kappa_{i}\right)$ that measures
the monetary cost of bringing the level of awareness of node $i$
to the level $\kappa_{i}$. Therefore, the total cost of our awareness
campaign is $C\triangleq\sum_{i=1}^{n}f_{i}\left(\kappa_{i}\right)$.

\subsection{The Cost of Disease Awareness}

We assume that we are only able to control the awareness level of
an individual within a certain interval $\underline{\kappa_{i}}\leq\kappa_{i}\leq\overline{\kappa_{i}}$.
Also, we assume the cost functions presents the following properties:
(\emph{a}) $f_{i}\left(\overline{\kappa_{i}}\right)=\overline{C_{i}}$,
i.e., $\overline{C_{i}}>0$ is the cost of bringing node $i$ to its
maximum level of awareness; and (\emph{b}) $f_{i}\left(\kappa_{i}\right)$
is nondecreasing for $x_{i}\in\left[\underline{\kappa_{i}},\overline{\kappa_{i}}\right]$.
In what follows, we develop an optimization framework to find the
cost-optimal distribution of awareness throughout the population to
control an epidemic outbreak.

\subsection{Convex Optimization Framework}

Theorem \ref{SAIS_Stab} provides a condition to control an epidemic
outbreak in a network of agents following the heterogeneous SAIS dynamics
in (\ref{dp}) and (\ref{dq}). Hence, the cost-optimal investment
on awareness can be found by solving the following optimization program:
\begin{align}
\min_{\left\{ \kappa_{i}\right\} } & \sum_{i}f_{i}\left(\kappa_{i}\right)\label{eq:MainOpt}\\
\mbox{s.t. } & \lambda_{1}(LBA_{\mathcal{G}}-MD)\leq0,\nonumber \\
 & \underline{\kappa_{i}}\leq\kappa_{i}\leq\overline{\kappa_{i}},\nonumber 
\end{align}
Notice that, if solved, this optimization program would find the cost-optimal
levels of awareness, $\left\{ \kappa_{i}^{*}\right\} _{i=1}^{n}$,
to control an epidemic outbreak. In what follows, we provide a computationally
efficient optimization framework to find the optimal investment profile
under certain assumptions on the cost function.

First, we show how to rewrite the spectral condition in (\ref{eq:MainOpt})
as a semidefinite constraint:
\begin{lem}
For $A_{\mathcal{G}}$ symmetric, we have that $\lambda_{1}(LBA_{\mathcal{G}}-MD)\leq0$,
if and only if
\begin{equation}
A_{\mathcal{G}}-\mbox{diag}\left\{ \frac{r_{i}\delta_{i}+\kappa_{i}\delta_{i}/\beta_{i}}{r_{i}\beta_{i}+r_{i}\kappa_{i}}\right\} \preceq0.\label{eq:SDPcondition}
\end{equation}
\end{lem}
\begin{IEEEproof}
Notice that $LBA_{\mathcal{G}}-MD$ is a matrix similar to $\left(LB\right)^{1/2}A_{\mathcal{G}}\left(LB\right)^{1/2}-MD$.
Since this matrix is symmetric, its eigenvalues are real. Then, we
have that $\lambda_{1}(LBA_{\mathcal{G}}-MD)\leq0$ if and only if
$\left(LB\right)^{1/2}A_{\mathcal{G}}\left(LB\right)^{1/2}-MD\preceq0$.
Then, we apply a congruence transformation by pre- and post-multiplying
by the diagonal matrix $\left(LB\right)^{-1/2}$ to obtain that $\left(LB\right)^{1/2}A_{\mathcal{G}}\left(LB\right)^{1/2}-MD\preceq0$
if and only if $A_{\mathcal{G}}-\left(LB\right)^{-1}MD\preceq0$.
(Notice that $L,B,M,$ and $D$ are all diagonal matrices.) From the
definitions of $L,B,M,D$ in Theorem \ref{SAIS_Stab}, we have that
\[
\left(LB\right)^{-1}MD=\mbox{diag}\left\{ \frac{r_{i}\delta_{i}+\kappa_{i}\delta_{i}/\beta_{i}}{r_{i}\beta_{i}+r_{i}\kappa_{i}}\right\} ,
\]
proving the statement of our Lemma.
\end{IEEEproof}
Using the above lemma, we can rewrite the optimization problem in
(\ref{eq:MainOpt}) as
\begin{align}
\min_{\left\{ \boldsymbol{\kappa}_{i}\right\} } & \sum_{i}f_{i}\left(\kappa_{i}\right)\label{eq:DiagOpt}\\
\mbox{s.t. } & A_{\mathcal{G}}-\mbox{diag}\left\{ \frac{r_{i}\delta_{i}+(\delta_{i}/\beta_{i})\kappa_{i}}{r_{i}\beta_{i}+r_{i}\kappa_{i}}\right\} \preceq0,\nonumber \\
 & \underline{\kappa_{i}}\leq\kappa_{i}\leq\overline{\kappa_{i}}.\nonumber 
\end{align}
We now prove that the above optimization program is quasiconvex. Define
the new set of variables $y_{i},\kappa_{i}$, as $y_{i}\triangleq\frac{r_{i}\delta_{i}+\kappa_{i}\delta_{i}/\beta_{i}}{r_{i}\beta_{i}+r_{i}\kappa_{i}},$
or equivalently,
\begin{equation}
\kappa_{i}=\frac{y_{i}r_{i}\beta_{i}-r_{i}\delta_{i}}{\delta_{i}/\beta_{i}-y_{i}r_{i}}\triangleq g_{i}\left(y_{i}\right).\label{eq:g_i}
\end{equation}
The function $g$ is a linear-fractional function and, therefore,
quasiconvex if $y_{i}\in\left\{ y\mbox{ s.t. }\delta_{i}/\beta_{i}-y_{i}r_{i}>0\right\} $,
\cite{Boy04}. This condition can be proved to always be true because
$r_{i}<1$ by definition. Hence, defining $Y=\mbox{diag}\{y_{i}\}$,
the optimization problem in (\ref{eq:DiagOpt}) can be written as
\begin{align*}
\min_{y_{i}} & \sum_{i}h_{i}\left(y_{i}\right)\\
\mbox{s.t. } & A_{\mathcal{G}}-Y\preceq0,\\
 & \underline{y_{i}}\leq y_{i}\leq\overline{y_{i}},
\end{align*}
where $h_{i}\left(y_{i}\right)\triangleq\left(f_{i}\circ g_{i}\right)\left(\kappa_{i}\right)$,
$\underline{y_{i}}\triangleq h(\overline{\kappa_{i}})$ and $\overline{y_{i}}\triangleq h(\underline{\kappa_{i}})$.
Since $f_{i}$ is a nondecreasing function and $g_{i}$ is a linear-fractional
transformation, the composition function $f_{i}\circ g_{i}$ is quasiconvex
\cite{Boy04}. This problem is not, in general, a semidefinite program
\cite{VB96}, due to the linear-fractional cost. It is, instead, quasiconvex
and can still be efficiently solved using, for example, the results
in \cite{BE93} (and references therein).

\subsection{SDP Formulation}

We can achieve a semidefinite formulation of our problem when the
cost function presents a particularly interesting structure. Assume
that the cost function is a linear-fractional function, as follows
\begin{equation}
f_{i}\left(\kappa_{i}\right)=\frac{c_{i}+s_{i}\kappa_{i}}{r_{i}\beta_{i}+r_{i}\kappa_{i}},\label{eq:LinearFracCost}
\end{equation}
where $r_{i},\beta_{i}$ are given parameters of the SAIS model, and
$c_{i},s_{i}$ are free parameters that can be adjusted to fit the
properties of the cost function. We are interested in nondecreasing
cost functions for which $f_{i}\left(\overline{\kappa_{i}}\right)=\overline{C_{i}}$.
These conditions are satisfied when $c_{i}$ and $s_{i}$ satisfy
$c_{i}=\overline{C_{i}}r_{i}\left(\beta_{i}+\overline{\kappa_{i}}\right)-s_{i}\overline{\kappa_{i}},$
and $s_{i}>\overline{C_{i}}/2.$ Moreover, if we want to satisfy $f_{i}\left(\underline{\kappa_{i}}\right)=0$,
we should simply choose $s_{i}=C_{i}r_{i}(\beta_{i}+\overline{\kappa_{i}})/(\overline{\kappa_{i}}-\underline{\kappa_{i}})$.

We can transform the optimization problem in (\ref{eq:DiagOpt}) with
the linear-fractional cost function (\ref{eq:LinearFracCost}) into
an SDP and efficiently solve it using standard optimization software
(such as cvx). In particular, let us perform the change of variables
\cite{CC62} 
\begin{align}
u_{i} & \triangleq\frac{\kappa_{i}}{r_{i}\beta_{i}+r_{i}\kappa_{i}},\label{eq:ui}\\
w_{i} & \triangleq\frac{1}{r_{i}\beta_{i}+r_{i}\kappa_{i}}.\label{eq:wi}
\end{align}
To avoid singularities in the transformation, we assume $r_{i}\beta_{i}+r_{i}\kappa_{i}>0$,
or equivalently $w_{i}>0$. Then, the objective function (\ref{eq:LinearFracCost})
and the elements of the diagonal matrix in the first constraint of
(\ref{eq:DiagOpt}) can be written as
\begin{align*}
\frac{c_{i}+s_{i}\kappa_{i}}{r_{i}\beta_{i}+r_{i}\kappa_{i}} & =c_{i}w_{i}+s_{i}u_{i},\\
\frac{r_{i}\delta_{i}+(\delta_{i}/\beta_{i})\kappa_{i}}{r_{i}\beta_{i}+r_{i}\kappa_{i}} & =r_{i}\delta_{i}w_{i}+(\delta_{i}/\beta_{i})u_{i},
\end{align*}
respectively, which are linear functions on the new variables $u_{i},w_{i}$.
To cast our problem as an SDP, we also need to rewrite the constraints
$\underline{\kappa_{i}}\leq\kappa_{i}\leq\overline{\kappa_{i}}$ in
terms of $u_{i},w_{i}$. We can do so by multiplying the three terms
in these inequalities $\underline{\kappa_{i}}\leq\kappa_{i}\leq\overline{\kappa_{i}}$
by $\left(r_{i}\beta_{i}+r_{i}\kappa_{i}\right)^{-1}$. Hence, we
obtain that $\underline{\kappa_{i}}w_{i}\leq u_{i}\leq\overline{x_{i}}w_{i}$,
which are linear inequalities on the new variables $u_{i},w_{i}$.

Finally, we obtain the SDP formulation for the optimization problem
(\ref{eq:DiagOpt}) with the linear-fractional cost function (\ref{eq:LinearFracCost})

\begin{align}
\min_{\left\{ u_{i},w_{i}\right\} } & \sum_{i=1}^{n}c_{i}w_{i}+s_{i}u_{i}\label{eq:SDP}\\
\mbox{s.t. } & A_{\mathcal{G}}-FW-GU\preceq0,\nonumber \\
\mbox{} & \underline{\kappa_{i}}w_{i}\leq u_{i}\leq\overline{\kappa_{i}}w_{i},\nonumber \\
\mbox{(a)} & w_{i}\geq0,\nonumber \\
\mbox{(b)} & r_{i}\beta_{i}w_{i}+r_{i}u_{i}=1.\nonumber 
\end{align}
where $U\triangleq\mbox{diag}\{u_{i}\}$ and $W\triangleq\mbox{diag}\{w_{i}\}$
are decision variables, and $F\triangleq\mbox{diag}\{r_{i}\delta_{i}\}$
and $G\triangleq\mbox{diag}\{\delta_{i}/\beta_{i}\}$ are matrices
of parameters. Notice that the constraint (a) imposes $r_{i}\beta_{i}+r_{i}\kappa_{i}>0$
in order to avoid singularities in our transformation, and the constraint
in (b) is required to guarantee that $r_{i}\beta_{i}w_{i}+r_{i}u_{i}=\frac{r_{i}\beta_{i}+r_{i}\kappa_{i}}{r_{i}\beta_{i}+r_{i}\kappa_{i}}=1.$

\section{Simulations}

We consider a small social network with $n=247$ nodes and adjacency
matrix $A_{FB}$, extracted from a real online social network. This
social network has $m=940$ edges and its largest eigenvalue is $\lambda_{1}=13.52$.
In our experiments, we are choosing a homogeneous recovery rate $\delta_{i}=\delta=1/7$.
We are also choosing a homogeneous infection rate $\beta_{i}=\beta=1.5\times\delta/\lambda_{1}=7.4e-3$.
(Notice that this recovery rate is not enough to control an epidemic
in the classical SIS model \cite{van2009TN}.) Awareness in the model
is quantified by two node-dependent parameters: $r_{i}$ and $\kappa_{i}$.
In our experiments, we choose a homogeneous $r_{i}=r=1/2$, i.e.,
aware individuals get infected at half the rate as ``unaware'' individuals. 

Our objective is to find the cost-optimal values for $\{\kappa_{i}\}_{i}$
to control an epidemic outbreak. Since $\kappa_{i}$ is the rate in
which susceptible individuals become aware of an epidemic, we can
interpret the value of $\kappa_{i}$ as the level of awareness of
node $i$. In our experiments, we assume that the level of awareness
can vary in the interval $\kappa_{i}\in\left[0,0.024\right]$. Moreover,
the cost of bringing an individual to a level of awareness $\kappa_{i}$
is given by $f_{i}$ in (\ref{eq:LinearFracCost}) where we have chosen
$c_{i}$ and $s_{i}$ such that $f_{i}$ satisfy both $f_{i}\left(\underline{\kappa_{i}}\right)=0$
and $f_{i}\left(\overline{\kappa_{i}}\right)=\overline{C_{i}}=1$.

Using these parameters, we run the optimization program in (\ref{eq:SDP})
using the standard software cvx. In Fig. 1, we present a scatter plot
with $247$ data points (as many as individuals in the network), where
each point has an abscissa equal to $f_{i}(\kappa_{i})$ (the investment
made on node $i$ to increase its level of awareness) and an ordinate
of $d_{i}$ (the degree of node $i$). We observe that there is a
strong dependence between the optimal level of investment in node
and its degree.

\begin{figure}
\includegraphics[scale=0.27]{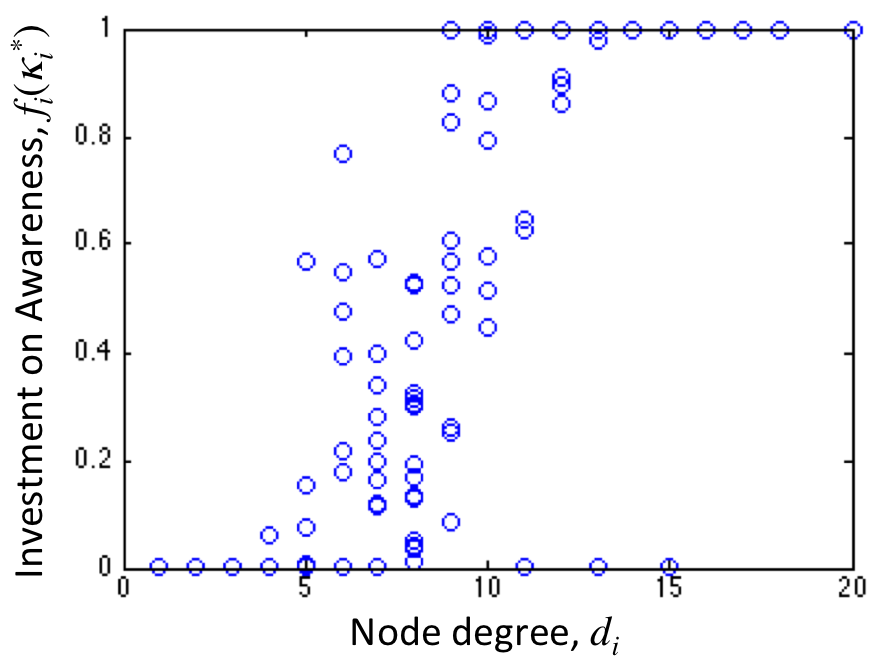}

\caption{Investment on awareness for each one of the 247 nodes in a social
network versus its degree.}
\end{figure}

\section{Conclusions}

We have proposed a convex optimization framework to find the cost-optimal
investment on awareness in a network of individuals. Our work is based
on the SAIS model \cite{FaryadCDC11SAIS}, which incorporates social-aware
response of individuals to epidemic spreading. We have derived a condition
to control the spread of an epidemic outbreak in terms of the eigenvalues
of a matrix that depends on the network structure and the parameters
of the model. We have then proposed an optimization program to find
the cost-optimal investment on disease awareness throughout the network.
Under mild assumptions on the cost function structure, we were able
to successfully cast this program into a convex optimization framework
allowing for an efficient solution.

\bibliographystyle{ieeetr}
\bibliography{ViralSpread}

\end{document}